\documentclass[prx,reprint,floatfix,letterpaper,twocolumn,nofootinbib,showpacs,longbibliography,accepted=2019-01-18]{quantumarticle}
\usepackage[caption=false]{subfig}
\usepackage{graphicx} % support the \includegraphics command and options
\usepackage{epstopdf}
\usepackage{array}
\usepackage{verbatim}
\usepackage{amsmath,amsfonts,amssymb,amscd}
\usepackage{amsthm}
\usepackage{tabularx}
\usepackage{stmaryrd}
\usepackage{enumerate}
\usepackage[ruled]{algorithm2e}
\usepackage{enumitem}
\usepackage{wasysym}
\usepackage[dvipsnames]{xcolor}
\usepackage[normalem]{ulem}
\usepackage{bbold}
\input{Qcircuit}

\usepackage{hyperref}

\theoremstyle{plain}
\newtheorem{thm}{Theorem}
\newtheorem{cor}[thm]{Corollary}
\newtheorem{lem}[thm]{Lemma}

\theoremstyle{definition}

\newcommand{\eq}[1]{(\hyperref[eq:#1]{\ref*{eq:#1}})}

\renewcommand{\sec}[1]{\hyperref[sec:#1]{Section~\ref*{sec:#1}}}
\newcommand{\thrm}[1]{\hyperref[thm:#1]{Theorem~\ref*{thm:#1}}}
\newcommand{\lemm}[1]{\hyperref[lemm:#1]{Lemma~\ref*{lemm:#1}}}
\newcommand{\prop}[1]{\hyperref[prop:#1]{Proposition~\ref*{prop:#1}}}
\newcommand{\corr}[1]{\hyperref[corr:#1]{Corollary~\ref*{corr:#1}}}
\newcommand{\fig}[1]{\hyperref[fig:#1]{Figure~\ref*{fig:#1}}}

\newcommand{\rank}[1]{\text{rank}(#1)}

\DeclareMathAlphabet{\matheu}{U}{eus}{m}{n}

\DeclareMathOperator{\im}{im}

\newcommand{\ket}[1]{|{#1}\rangle}

\newcolumntype{L}[1]{>{\raggedright}p{#1}}
\newcolumntype{C}[1]{>{\centering}p{#1}}
\newcolumntype{R}[1]{>{\raggedleft}p{#1}}
\newcolumntype{D}{>{\centering\arraybackslash}X}
\definecolor{darkgreen}{rgb}{0,0.5,0}
\definecolor{darkblue}{rgb}{0,0,0.5}

\begin{document}
\title{Fault-tolerant gates via homological product codes}
\author{Tomas Jochym-O'Connor}
\email{tjoc@caltech.edu}
\affiliation{\small Walter Burke Institute for Theoretical Physics,
Institute for Quantum Information \& Matter, California Institute of Technology, Pasadena, CA 91125, USA}

\begin{abstract}
A method for the implementation of a universal set of fault-tolerant logical gates is presented using homological product codes. In particular, it is shown that one can fault-tolerantly map between different encoded representations of a given logical state, enabling the application of different classes of transversal gates belonging to the underlying quantum codes. This allows for the circumvention of no-go results pertaining to universal sets of transversal gates and provides a general scheme for fault-tolerant computation while keeping the stabilizer generators of the code sparse.
\end{abstract}

\maketitle

\section{Introduction}

Quantum error correction extends qubit coherence times through error mitigation and will be a requirement for any large-scale quantum computation. In this vein, tremendous research efforts have been placed on finding quantum error correcting codes that may be realizable in both the near and distant future. Among the leading candidates for experimental implementation are 2D~topological stabilizer codes, such as the toric code~\cite{Kitaev03,FMMC12}, which allow for the correction of errors by measuring small-weight local checks while protecting logical information in highly non-local degrees of freedom~\cite{BKM+14,CGM+14,NMM+14,KBF+15,TCM+16}. These codes are experimentally appealing due to their stabilizers being low-weight, thus minimizing the effect of noise during measurement. They can be generalized to higher spatial dimensions, again with the stabilizer generators being relatively low-weight, and provide theoretically simple implementations of different classes of fault-tolerant logic~\cite{Bombin15a, KB15, BK13, PY15, KYP15,JOKY18}. The motivation for quantum error correcting codes to have geometrically local stabilizers (in a given dimension) is to simplify experimental architectures, yet this may not necessarily be a hard requirement. However, the need for low-weight stabilizer checks is much stronger, as larger weight checks lead to more noise propagation, and generally lower threshold error rates. The theory of quantum low-density parity check~(LDPC) codes has been developed to address such concerns and finding good codes with low-weight checks remains a very active area of research~\cite{MMM04,KP12,BrHa14,FH14,TZ14,Hastings17}. Moreover, LDPC constructions can be used to construct codes with very little overhead for fault-tolerant computation, relying on the preparation of logical ancillary states~\cite{Gottesman14,LTZ15,FGL18,FGL18b}. This work will not focus on the development of such codes, but rather will center on how to generally use such codes for the purposes of fault-tolerant computation.

Obtaining a universal set of fault-tolerant gates is complicated by the existence of no-go results for such constructions using only transversal gates~\cite{EK09,ZCC11}. However, many alternative schemes have been developed to circumvent this restriction. They rely on the preparation of special ancillary states and gate teleportation~\cite{Shor96,KLZ96,BK05}, or tailored fault-tolerant constructions for certain classes of codes~\cite{PR13, Bombin15a, JL14, ADP14, Bombin15c, BC15,JB16,JBH16,YTC16}. This work extends the set of fault-tolerant alternatives, presenting a scheme for fault-tolerant logic on any CSS code while keeping the underlying stabilizer measurements low-weight.

We present a method for implementing fault-tolerant logical gates in a homological product code~\cite{BrHa14}. Namely, given the homological product of two quantum codes, we show how to map in a fault-tolerant manner between the encoded homological product logical state to a logical state specified by one of the two codes. Then, if the underlying codes have a set of transversal gates, such logical gates can be applied and the state can be re-encoded back into the full codespace fault-tolerantly. There are no restrictions on the underlying codes, other than having to be defined by a boundary operator~$\delta:\mathcal{C} \rightarrow \mathcal{C}$ such that~$\delta^2 = 0$ in the linear space~$\mathcal{C}$. In particular, by using versions of the 2D and 3D~color codes~\cite{BM07,KB15} as the underlying codes in the construction, a universal set of fault-tolerant operations can be implemented. The mapping between different representations of the code follows a similar construction to Ref.~\cite{JL14}, where one of the two codes is unencoded while the other provides the protection. The key difference between the methods is that the stabilizers of the homological product can remain low-weight, unlike those in a concatenated model. Moreover, we present a decoding method to address errors that may occur during the unencoding of one of the two codes. Given certain properties of the underlying codes, this will result in a finite probability error threshold along the lines of Ref.~\cite{KP13} as well as potential protection against measurement errors.

The article is organized as follows: In Sec.~\ref{sec:HomologicalPC} we review the theory of CSS codes defined by chain complexes and the construction of the homological product codes, carefully considering their underlying structure. In Sec.~\ref{sec:FTlogic} we present the main result, a fault-tolerant method to implement a logical gate using homological product codes as well as discuss a simple decoding procedure. In Sec.~\ref{sec:Universal} we present examples of codes that exhibit a set of universal fault-tolerant gates, expanding on the notions of code padding and doubling. Finally, we conclude with some remarks and open questions in Sec.~\ref{sec:Conclusion}.

\section{Homological Product Codes}
\label{sec:HomologicalPC}
\subsection{Single sector theory}
\label{sec:singleSec}

We begin by reviewing the connection between CSS codes and homology. Namely, as in Ref.~\cite{BrHa14}, we focus on single sector theory\footnote{In general, CSS codes can be constructed from chain complexes where $\delta_i:\mathcal{C}_i \rightarrow \mathcal{C}_{i-1}$ are linear operator satisfying: $\delta_{i-1}\delta_i = 0$.} in~$\mathbb{Z}_2$, that is a chain complex defined by a linear space~$\mathcal{C}$ and a linear boundary operator~$\delta: \mathcal{C} \rightarrow \mathcal{C}$, such that~$\delta^2 = 0$. We can then use such a boundary operator to define a CSS code~\cite{CS96,Steane96b}, that is a stabilizer code whose generators can be expressed as either $X$-type or $Z$-type.

Let $\mathcal{C}$ be a $n$-dimensional binary space~$\mathbb{Z}_2^n$, then $\delta$ will be a $n \times n$ binary matrix. The (perhaps over-complete) generating set of $X$~stabilizers will be given by the rows of~$\delta$, that is for a given row, a generator~$S_{X_i}$ will have $X$~support on the qubits with a~1 in the given row. Similarly, the $Z$~stabilizers will be defined by the columns of the matrix. Given $\delta^2=0$, we are thus assured commutativity of the stabilizers. The number of independent generators of both type will be~$\rank{\delta}$, and as such the number of logical qubits of the code will be~$k = n - 2 \rank{\delta}$.

For the remainder of this section we shall focus on the reverse implication. That is, given a CSS code whose $X$ and~$Z$ stabilizer spaces are of the same dimension, one can construct a boundary operator~$\delta$ of a single-sector theory. The results presented are a fairly straightforward corollary of Lemma~3 from Ref.~\cite{BrHa14}, yet we include them here for completeness and to review some important concepts, namely the canonical boundary operator.

\begin{lem}
\label{lem:canonical}
Given a CSS code on $n$ qubits, whose $X$ and $Z$~stabilizer spaces are each of cardinality~$2^l$, therefore encoding~$k = n-2l$ qubits. Then, there exists a invertible matrix~$W$, and canonical boundary operator~$\delta_0$ defined as:
\begin{align}
\delta_0 = \begin{pmatrix}
0_k & 0 & 0 \\
0 & 0_l & \mathbb{1}_l \\
0 & 0_l & 0_l
\end{pmatrix},
\end{align}
such that $\delta  = W \delta_0 W^{-1}$, where the rows~(columns) of~$\delta$ contain a set of generators of the $X$~($Z$)~stabilizer group.
\end{lem} 

\begin{proof}
Given the existence of a CSS code, by the Gottesman-Knill theorem~\cite{Gottesman97} there exists a unitary operator~$U$ composed solely of CNOT gates that maps~$\ket{\psi} \otimes \ket{0}^{\otimes l} \otimes \ket{+}^{\otimes l}$ to the encoded stabilizer code, where~$l = (n-k)/2$ and $\ket{\psi}$ is a $k$-qubit state. This statement can be expressed in terms of matrix manipulation on~$\mathbb{Z}_2$, where $\delta_0$~will represent the initial state of the stabilizers before the application of the encoding circuit, that is the rows of~$\delta_0$ represent the initial~$\ket{+}$ states, and the columns the initial~$\ket{0}$ states. A CNOT gate with qubit~$i$ as control, and qubit~$j$ as target can then be expressed according to the invertible matrix:
\begin{align}
W_{i,j} = \mathbb{1} + w_iw_j^{T},
\end{align}
where $w_k$ is the standard basis column vector one non-zero entry at position~$k$. The action of~$W_{i,j}$ by conjugation maps column~$c_j$ to the sum of columns~$c_i \oplus c_j$, and row $r_i$ to the sum of rows~$r_i \oplus r_j$. This is the exact action required from a CNOT, as it maps~$X_i$ to $X_iX_j$ and $Z_j$ to $Z_iZ_j$. Note, as required for a valid representation of a CNOT gate, $W_{i,j}^2 = \mathbb{1} \Rightarrow W_{i,j} = W_{i,j}^{-1}$, where again we are working modulo 2. Then, the encoding operation~$W$ can be broken into its CNOT components and can be expressed as~$W = W_{i_N,j_N}\cdots W_{i_2,j_2}W_{i_1,j_1}$. As such,
\begin{align}
\delta &= (W_{i_N,j_N}\cdots W_{i_1,j_1}) \delta_0 (W_{i_N,j_N}\cdots W_{i_1,j_1})^{-1}\\
&= (W_{i_N,j_N}\cdots W_{i_1,j_1}) \delta_0 (W_{i_1,j_1} \cdots W_{i_N,j_N}),
\end{align}
will be a valid representation of the stabilizers of the code. Since $W$~is a valid representation of the encoding circuit of the code, the rows (columns) of $\delta$ will remain a valid representation of the $X$ ($Z$)~stabilizers since they were so for~$\delta_0$.
\end{proof}

Perhaps as importantly for the purposes of this article, if the given CSS code has a generating set of stabilizers that are sparse, then the resulting constructed~$\delta$ will be sparse. We define the \emph{sparsity} of a code to be the smallest integer~$w$ such that there exists a set of generators of the code whose weights are at most~$w$ while any given qubit participates in at most~$w$ stabilizer checks.

\begin{cor}
\label{cor:sparse}
Given a CSS code on $n$ qubits with an equal number of $X$ and $Z$~stabilizers and sparsity~$t$, then a boundary operator~$\delta$ can be constructed such that no row or column will have more than~$t^2$ non-zero entries.
\end{cor}

\begin{proof}
Given some sparse representative set of stabilizers~$\{S_{X_i}, S_{Z_i} \}$, as in the proof of Lemma~\ref{lem:canonical}, a unitary circuit~$W$ can be chosen that maps $Z_{k+i} \rightarrow S_{Z_i}$ and $X_{k+l+i} \rightarrow S_{X_i}$. Consider the right action of $W^{-1}$ in terms of its action on the canonical boundary operator:
\begin{align}
\delta_0W^{-1} = \begin{pmatrix}
0_{k \times n} \\
s_{X_1} \\
\vdots \\
s_{X_l} \\
0_{l \times n} 
\end{pmatrix},
\end{align}
where on the right side of the equality we have a matrix whose rows are either all-zero or a binary representation~$s_{X_i}$ of the stabilizer~$S_{X_i}$. This follows from the fact that the right application of~$W^{-1}$ results in the propagation of the initial~$X$ stabilizers to their final generator form. Then, by the sparsity of the stabilizer generators, each row will be of weight at most~$t$ and each column will have weight at most~$t$. Now consider the left application of $W$ applied to~$\delta_0W^{-1}$, thus completing the conjugation, the resulting matrix $W\delta_0 W^{-1}$ will have rows that will be sums of the different rows of~$\delta_0W^{-1}$. Moreover, each row of~$W\delta_0 W^{-1}$ will be a sum of at most~$t$ rows~$s_{X_i}$, and will as such be of weight at most~$t^2$. Finally, since a given row can map to at most~$t$ other rows, each non-zero entry within a column of~$\delta_0W^{-1}$ can map to at most $t$~entries within that column. Therefore, since there were at most~$w$ non-zeros entries in a column of~$\delta_0W^{-1}$, there can be at most~$t^2$ non-zeros in each column of~$W\delta_0 W^{-1}$
\end{proof}

\subsection{Homological Product Construction}
\label{sec:HomologialProd}

Given two complexes~$(\mathcal{C}_1,\delta_1)$, $(\mathcal{C}_2,\delta_2)$ with their associated spaces and single sector boundary operators, we define a new operator as in Ref.~\cite{BrHa14},
\begin{align}
\partial = \delta_1 \otimes \mathbb{1} + \mathbb{1} \otimes \delta_2,
\end{align}
acting on~$\mathcal{C}_1 \otimes \mathcal{C}_2$. It follows from~$\delta_i^2=0$ that~$\partial^2 = 0$, again since we are working in~$\mathbb{Z}_2$. Therefore, $(\mathcal{C}_1 \otimes \mathcal{C}_2, \partial)$ is a valid single sector complex, defining its own quantum CSS code. 

We now restate some important properties of the homological product.
\begin{lem}[\cite{BrHa14}]
\label{lem:kernel}
Let $(\mathcal{C}_1,\delta_1)$, $(\mathcal{C}_2,\delta_2)$ be complexes defining codes with $k_1$, and $k_2$~logical operators, respectively. Let $\partial = \delta_1 \otimes \mathbb{1} + \mathbb{1} \otimes \delta_2$, then the resulting complex~$(\mathcal{C}_1 \otimes \mathcal{C}_2, \partial)$ will encode~$k = k_1k_2$ logical qubits and
\begin{align}
\ker{\partial} = \ker{\delta_1} \otimes \ker{\delta_2} + \im{\partial}.
\end{align}
\end{lem}

Suppose that $w_a$ is the sparsity of~$\delta_a$. Moreover, let $d_a^X, d_a^Z$ be the $X$ and $Z$ distances for the corresponding codes. Then, the weight and distances of the new code can be bounded according to the parameters of the original code.

\begin{lem}[\cite{BrHa14}]
\label{lem:distance}
The sparsity of $\partial$ is upper bounded by~$w_1 + w_2$. The $X$ and $Z$~distance of the new code satisfy the following bounds:
\begin{align}
\max{\{d_1^\alpha,d_2^\alpha\}} \le d^\alpha \le d_1^{\alpha}d_2^{\alpha}, \qquad \alpha = X,Z.
\end{align}
\end{lem}

\subsection{Encoding the homological product code}

In this subsection, we review some facts about the encoding circuit for~$\partial$~\cite{BrHa14}. As eluded to in Section~\ref{sec:singleSec}, there are invertible matrices $W_a$ such that~$\delta_a = W_a \delta_{a,0} W_a^{-1}$, where~$\delta_{a,0}$ are the canonical boundary operators for~$\delta_a$. The matrices~$W_a$ are binary representatives of the encoding circuit for the given code, and as such, by taking their tensor product we obtain the encoding operation for~$\partial$. That is:
\begin{align}
\partial &= (W_1 \otimes W_2) (\delta_{1,0} \otimes \mathbb{1} + \mathbb{1} \otimes \delta_{2,0}) (W_1^{-1} \otimes W_2^{-1}) \\
&:= (W_1 \otimes W_2) \partial_0 (W_1 \otimes W_2)^{-1},
\end{align}
where we have defined~$\partial_0$ to be the canonical boundary operator for~$\partial$. We can express~$\partial_0$ in matrix form as follows: 
\begin{align}
\partial_0 &= (\delta_{1,0}\otimes\mathbb{1}+\mathbb{1}\otimes\delta_{2,0}) \\
&=\begin{pmatrix}
\mathbb{1}_{k_1} \otimes \delta_{2,0} & 0 & 0 \\
0 & \mathbb{1}_{l_1}\otimes \delta_{2,0} & \mathbb{1}_{l_1} \otimes \mathbb{1}_{n_2} \\
0 & 0 & \mathbb{1}_{l_1}\otimes \delta_{2,0}
\end{pmatrix},
\label{eq:delta0matrix}
\end{align}
where $k_i$ are the number of logical qubits and $l_i = (n_i - k_i)/2$ is the number of $X/Z$ stabilizers of the given code code. 

\begin{figure}[h]
\begin{center}
\includegraphics[width = \linewidth]{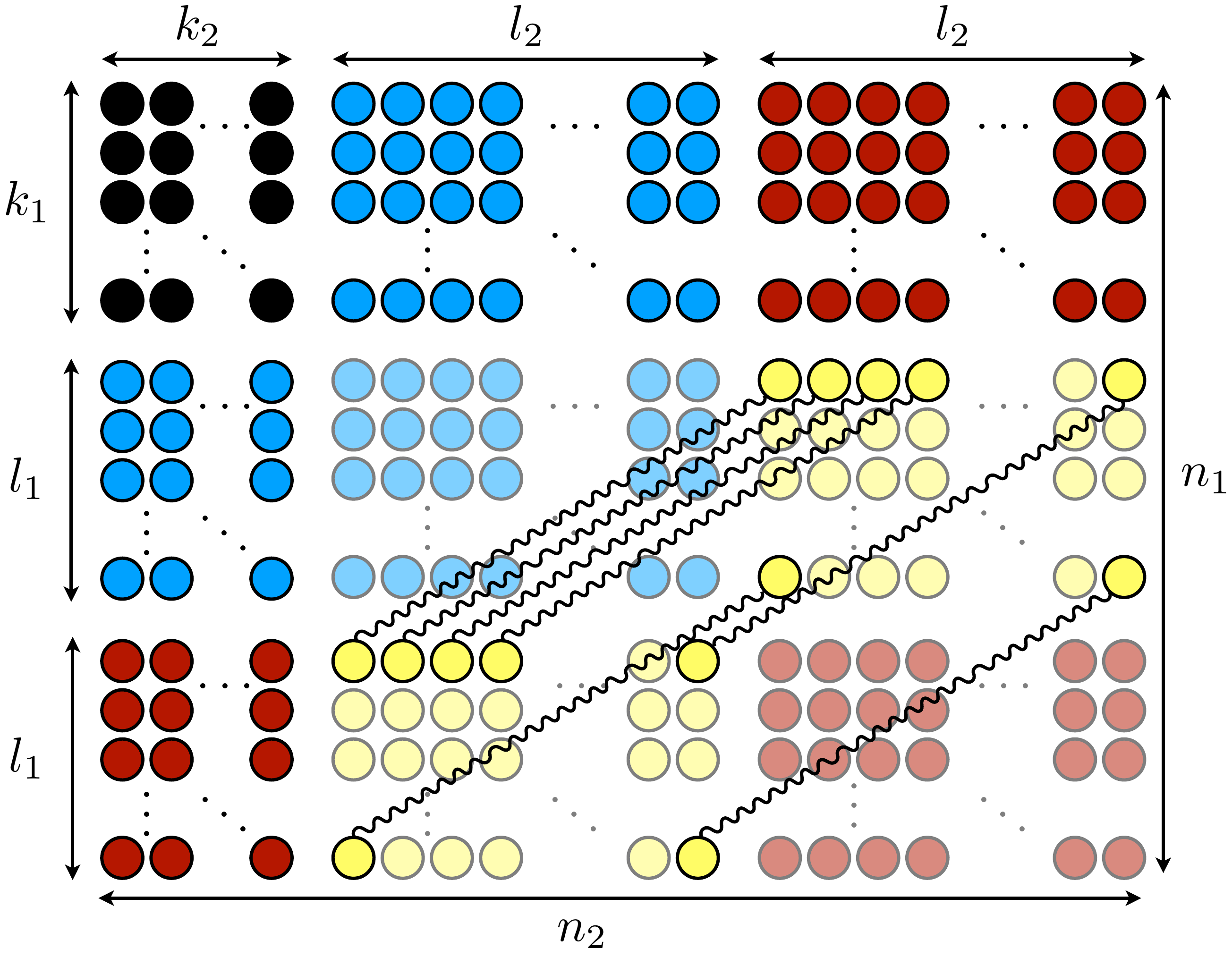}
\caption{Initial state of the homological product code prior to encoding~\cite{BrHa14}. Each circle represents a qubit, with the black qubits representing those holding the logical information to be encoded. Blue qubits are prepared in the~$\ket{0}$ state, while red qubits are prepared in~$\ket{+}$. The yellow qubits joined by an oscillating edge are prepared in a Bell pair~$(\ket{00}+\ket{11})/\sqrt{2}$.}
\label{fig:init}
\end{center}
\end{figure}

It is worth further exploring the form of~$\partial_0$, as this will be informative of how the logical information is encoded in the code. Each row and column of~$\partial_0$ will be of weight at most 2. Moreover, if a given row has 2 non-zeros entries, say at positions $q_i$ and $q_j$, then any column with a non-zero entry at~$q_i$ will also have a non-zero entry at position~$q_j$ in order to satisfy commutativity. As such, these rows and columns represent an initial entangled Bell pair between qubits~$q_i$ and~$q_j$ since they will be stabilized by the operators~$X_{q_i}X_{q_j}$ and~$Z_{q_i}Z_{q_j}$.

\begin{figure}
\subfloat[$\mathcal{C}_1$ encoder]{
\begin{minipage}[c][1\width]{0.5\linewidth}%
\includegraphics[clip,width=0.9\textwidth]{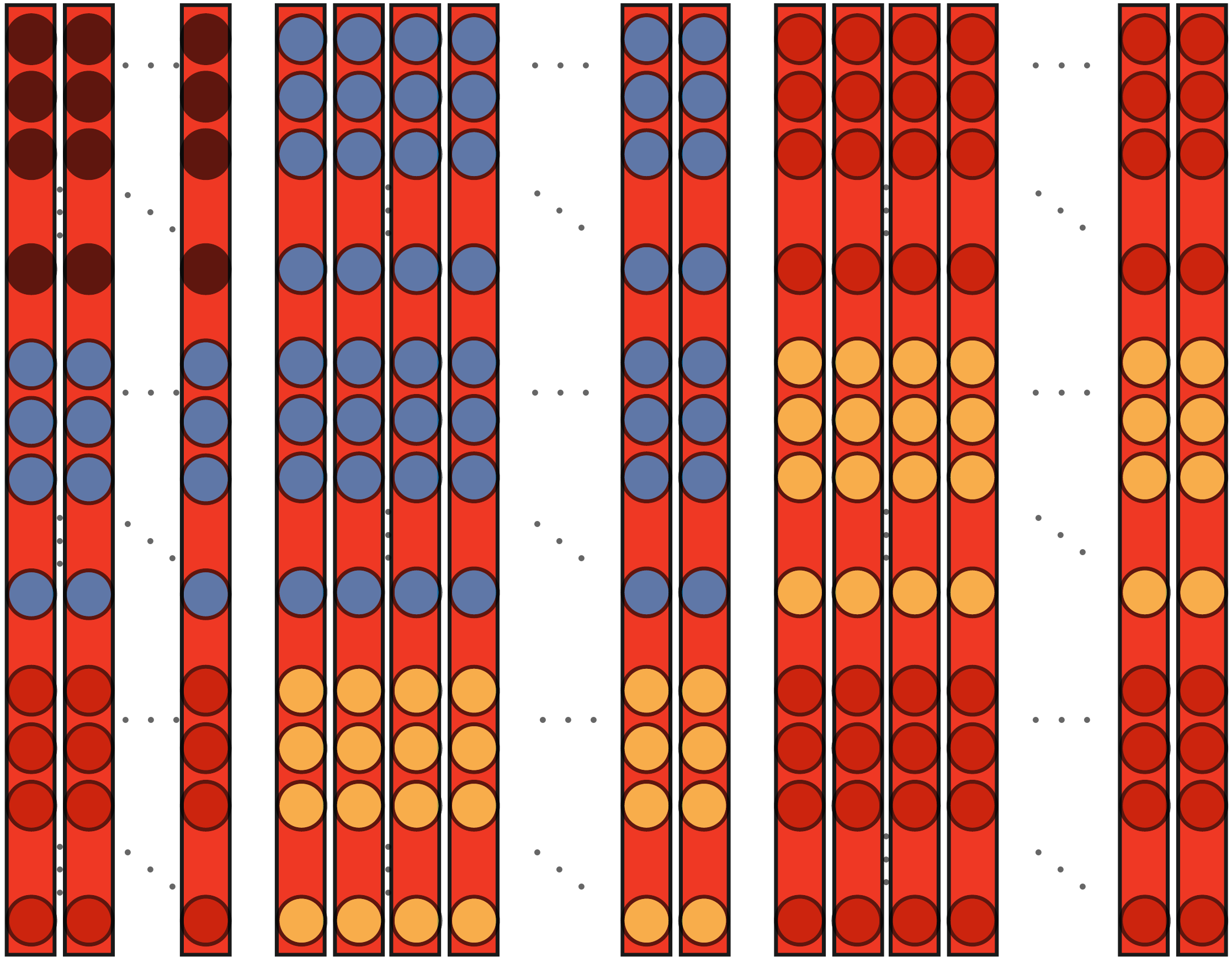}%
\label{fig:Uenc1}
\end{minipage}}
\subfloat[$\mathcal{C}_2$ encoder]{\centering{}%
\begin{minipage}[c][1\width]{0.5\linewidth}%
\begin{center}
\includegraphics[clip,width=0.9\textwidth]{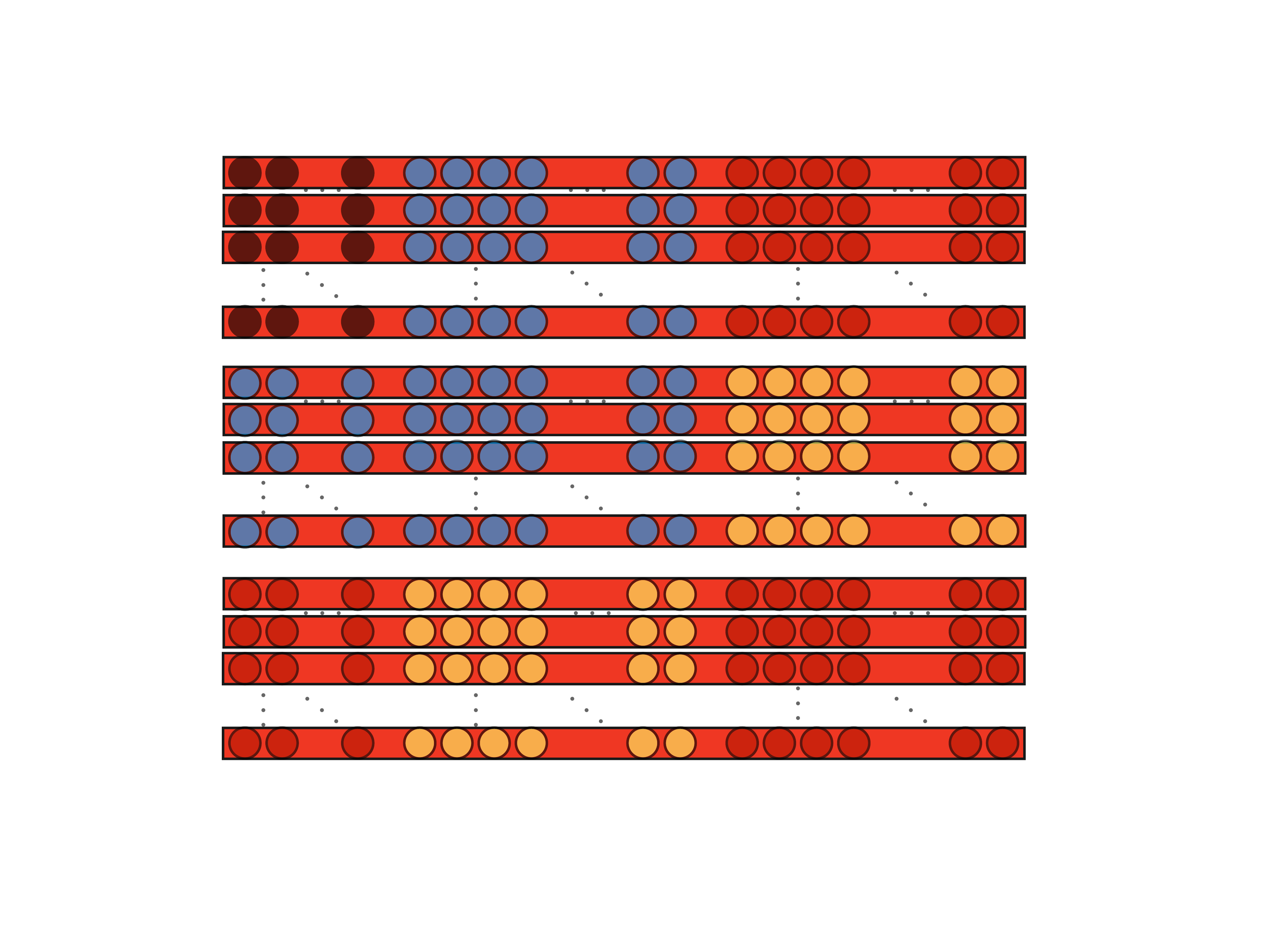}
\label{fig:Uenc2}
\par\end{center}%
\end{minipage}}
\caption{Schematic representation of the support of the encoding circuits for each of the codes underlying the homological product code, defined by the boundary operator~$\partial = \delta_1\otimes \mathbb{1} + \mathbb{1} \otimes \delta_2$. The overall encoding circuit has a binary representation of the operator~$W_1 \otimes W_2$, which acts on the initial state represented in Fig.~\ref{fig:init}. That is, the physical encoding unitary $U_1$ for the code~$\mathcal{C}_1$ will act on every column on qubits from Fig.~\ref{fig:init}, and conversely the physical encoding unitary~$U_2$ will act on every row.}
\label{fig:Uenc}
\end{figure}

The initial state can be pictorially represented by Fig.~\ref{fig:init}, where along a fixed row and column, the states in~$\mathcal{C}_1$ and~$\mathcal{C}_2$ are fixed, respectively. Then, the tensor product binary operators~$W_1 \otimes \mathbb{1}$ and~$\mathbb{1} \otimes W_2$ will have geometric meaning in this picture. Note for the remainder of this work, we will denote~$U_i$ as the physical encoding unitaries composed of CNOT gates acting on the quantum states whose binary representation is given by~$W_i$. Thus $U_1$ will couple qubits within vertical bands, while $U_2$ will couple qubits within horizontal bands, as represented in Fig.~\ref{fig:Uenc}.

\section{Fault-tolerant logical gates}
\label{sec:FTlogic}
\subsection{Partial decoding of the homological product code}
\label{sec:partial}

The key idea for expanding the set of available fault-tolerant logical gates will be for the two underlying codes composing the homological product to have different complimentary sets of transversal gates. Then, we can achieve the application of these logical gates by only decoding one of the two underlying codes, while remaining protected by the other. This is reminiscent of the scheme for implementing fault-tolerant gates using two concatenated codes~\cite{JL14}, with the added advantage that the stabilizers remain low-weight in the case of the homological product.

The main result is that, while we decode one of the codes, we still remain fully protected by the other code. While errors may potentially propagate during the application of the decoding process, they can still be corrected as long as the number of faults is less than half the distance of the underlying code protecting the information (that is the code that remains encoded at all times). The main Theorem is a variant of Lemma~\ref{lem:distance} stated above (originally Lemma~2 from Ref.~\cite{BrHa14}), yet is proved using the concept of error bands.

Recall the homological product is defined by the complex~$(\mathcal{C}_1 \otimes \mathcal{C}_2, \partial)$, where $\mathcal{C}_i$ are binary spaces. The complex defines a code with~$n_1n_2$~physical qubits, which we label $(i,j)$ where~$1 \le i \le n_1$ and $1 \le j \le n_2$. Then, a Pauli operator~$P$ is supported only on~$(i,j)$ if $\text{Tr}_{(i',j') \ne (i,j)} (P)= P_{i,j} \ne 0$, we will therefore denote such Pauli operators~$P_{i,j}$.  We will call the \emph{error band}~$E_a^1$ to be all possible Pauli operators of the form~$\prod_{j} P_{a,j} $, that is a product of any Pauli operators~$P_{i,j}$ with $i = a$. Conversely, the error band~$E_a^2$ will be all possible Pauli operators of the form~$\prod_j P_{j,a}$. Moreover, we will say that a Pauli error~$E$ is supported on the set~$\{ e_1, \cdots e_l \}$ of error bands~$E_a^1$ if the operator is supported by the qubits composing those error bands, that is~$E \subseteq E_{e_1}^1 \cup \cdots \cup  E_{e_l}^1$. A Pauli operator supported on a set of error bands~$E_a^2$ is defined similarly and whether the operator is supported on error bands~$E_a^1$ or~$E_a^2$ should be clear from context. Note that the unitary $U_1$ will only couple qubits within a fixed error bands~$E_a^1$, while conversely $U_2$ will only couple qubits within fixed error bands~$E_a^2$.

\begin{thm}
\label{thm:errorbands}
Let $(\mathcal{C}_1,\delta_1)$, $(\mathcal{C}_2,\delta_2)$ be complexes and let $(\mathcal{C}_1 \otimes \mathcal{C}_2, \partial)$ be the homological product code constructed from these codes with $\partial = \delta_1 \otimes \mathbb{1} + \mathbb{1} \otimes \delta_2$. Then any error~$E$ supported on fewer than $d_2$ error bands~$E_a^1$ cannot support a logical operator. Similarly, any error~$E$ supported on fewer than $d_1$ error bands~$E_a^2$ cannot support a logical operator.
\end{thm}

\begin{proof}
Consider an error~$E$ supported on fewer than~$d_2$ error bands~$E_a^1$, let the affected bands be denoted by the set~$\{e_1, \cdots, e_l \}$, that is~$E \subseteq E_{e_1}^1 \cup \cdots \cup  E_{e_l}^1$, where~$l<d_2$. Since any error can be expressed in the Pauli basis, if we show any Pauli error supported on the above set cannot support a logical operator, $E$ cannot support a logical operator. As such, without loss of generality, suppose $E$ is a Pauli error. Consider the modified error~$E' = U_{1}^{\dagger} E U_{1}$, where $U_{1}$ is the encoding unitary for the code~$\mathcal{C}_1$. Then, if we can show that $E'$ cannot support a logical operator on the new codespace after applying~$U_{1}^{\dagger}$, then by unitary equivalence it cannot on the homological product codespace. 

Any logical operator must commute with all of the stabilizers of the code. The modified codespace is given by the boundary operator~$\partial_{1,0} = (W_1^{-1} \otimes \mathbb{1}) \partial (W_1 \otimes \mathbb{1}) = \delta_{1,0} \otimes \mathbb{1} + \mathbb{1} \otimes \delta_2$, that is it will correspond to $k_1$ logical codeblocks encoded in the code~$( \mathcal{C}_2, \delta_2)$ along with accompanying encoded ancilla state\footnote{An encoded ancillary state is a fixed state that contains no non-trivial logical information, yet still may be partially encoded.}. This can be viewed visually by considering the initial unencoded state in Fig.~\ref{fig:init} followed by the encoding operation of Fig.~\ref{fig:Uenc2}. Therefore, in order for $E'$~to support a logical error, it will have to support a logical operator on one of the first~$k_1$ error bands~$E_a^2$. Without loss of generality, consider the first error band~$E_1^2$, that is the first row of qubits in Figs.~\ref{fig:init}-\ref{fig:Uenc}, and the support of~$E'$ on that band, $E_1' = E'\cap E_1^2 \subseteq (E_{e_1}^1 \cup \cdots \cup  E_{e_l}^1) \cap E_1^2$. Therefore, the weight of the Pauli operator given by~$E_1'$ is limited by the number of initial error bands~$E_a^1$ on which the error was supported, that is: $\text{wt}(E_1') \le l < d_2$, and as such since any logical operator supported on the band~$E_1^2$ must be of weight at least~$d_2$, the error~$E_1'$ cannot support a logical error. Since this will be true for all encoded logical bands supported on~$E_a^2$, with~$a\le k_1$, the error~$E'$ cannot support a non-trivial logical error. 

To conclude, since $E'$ cannot support a logical error on the code specified by the boundary operator~$\partial_{1,0}$, then $ E = U_{1}E'U_{1}^{\dagger}$ cannot support a logical operator on the code specified by~$\partial = (W_1 \otimes \mathbb{1}) \partial_{1,0} (W_1^{-1} \otimes \mathbb{1})$. 
\end{proof}

Equipped with Theorem~\ref{thm:errorbands}, we propose the following scheme to implement a fault-tolerant logical gate. Suppose the logical gate~$G_1$ can be implemented transversally on the single-qubit partition of the code~induced by the complex~$(\mathcal{C}_1,\delta_1)$, that is it can be implemented by applying gates that are each individually supported on single qubits of the code. Then, in order to apply the logical gate~$G_1$ on the logical state of the homological code~$(\mathcal{C}_1\otimes \mathcal{C}_2,\partial)$, we begin by unencoding the code~$(\mathcal{C}_2,\delta_2)$, that is we apply the unitary~$\mathbb{1} \otimes U_{2}^{\dagger}$. At this point, the first~$k_2$ blocks of~$n_1$ qubits remain encoded in the code~$(\mathcal{C}_1,\delta_1)$, while the the remaining blocks are in an encoded ancillary state. Therefore, we can apply the transversal implementation of the logical gate~$G_1$ on any of the $k_2$ logical states we desire. We complete the logical gate application by then reencoding into the code~$(\mathcal{C}_1\otimes \mathcal{C}_2,\partial)$ by applying~$\mathbb{1} \otimes U_2$.

The proposed scheme is fault-tolerant in that, it will be able to correct against up to~$\lfloor (d_1-1)/2 \rfloor$ faults throughout the process. Any error $P_{a,b}$ that occurs during the application of either~$\mathbb{1} \otimes U_2^{-1}$, $\mathbb{1} \otimes U_2$, or the transversal gate can spread to a high-weight error, yet such an error will remain within a single error band~$E_a^2$. This follows as the application of $\mathbb{1} \otimes U_2$ only ever couples qubits within the same error band~$E_a^2$. Similarly, a transversal gate with respect to the code~$(\mathcal{C}_1,\delta_1)$ will not couple different error bands~$E_a^2$, by definition. Therefore, any single fault~$P_{a,b}$ can result in an error that is contained within the error band~$E_{a}^2$. Since any logical error must be supported on at least $d_1$ such error bands, any error affecting less than half of such error bands must remain correctible by the Knill-Laflamme condition~\cite{KL97}.

By symmetry, given a transversal gate~$G_2$ on the single-qubit partition of the code~$(\mathcal{C}_2,\delta_2)$, a fault-tolerant implementation of~$G_2$ can be achieve by applying~$U_{1}^{\dagger} \otimes \mathbb{1}$, followed by the transversal gate, and a reencoding~$U_{1} \otimes \mathbb{1}$. Such a fault-tolerant gate will be able to correct against up to~$\lfloor (d_2-1)/2 \rfloor$ faults.

\subsection{Correcting errors}

In the last section, we showed how even when decoding one of the two codes, we can always protect against at least~$\lfloor (d_i-1)/2 \rfloor$~faults. However, the encoding/decoding operations may generally have~$\mathcal{O}(n_i^c)$ time steps\footnote{Decoding of stabilizer codes amounts to Gaussian elimination on the binary-symplectic matrices representing the stabilizers. This can be done using a polynomial number (in the input size) of elementary matrix row manipulations which correspond to Clifford gates.}, and as such, errors will accumulate within a given error band (assuming an independent non-Markovian error processes), resulting in an error with probability: $p_e\mathcal{O}(n_i^c)$, where $p_e$~is the physical error rate. This is undesirable from the perspective of fault tolerance, as we hope that for a given family of codes, by growing the distance, the probability of incurring a logical error decreases exponentially below some threshold value. Yet, if the underlying error rate is growing polynomially with the distance, this yields a pseudo-threshold for each code, rather than a global threshold for a code family. 

In this section, we present a simple decoding algorithm for the homological product code, based on the decoding algorithm of the individual codes composing the homological product. While the presented scheme will likely be far from ideal in many settings, it will serve as a proof of principle decoder as well as provide a means to correct against errors during the implementation of the fault-tolerant logical gates. This will help alleviate the concern errors accumulating within an error band due to the encoding/decoding having a macroscopic number of individual time steps. As will be discussed at the end of this subsection, in some cases this will guarantee a fault-tolerance threshold against independent noise. However, the existence and value of such a threshold would have to be studied on a case-by-case basis. 

Consider the homological product code as specified in the previous subsection, with a boundary operator~$\partial = \delta_1 \otimes \mathbb{1} + \mathbb{1} \otimes \delta_2$. Moreover, suppose we have recovery operators~$\mathcal{R}_i$ for each code that returns the code to the codespace and corrects with certainty when the weight of the error is below half the distance of the respective code. We present the following Corollary to Theorem~\ref{thm:errorbands}, which follows directly from the proof of that result.

\begin{cor}[of Theorem~\ref{thm:errorbands}]
\label{cor:errorbands}
Let $(\mathcal{C}_1,\delta_1)$, $(\mathcal{C}_2,\delta_2)$ be complexes and let $(\mathcal{C}_1 \otimes \mathcal{C}_2, \partial)$ be the homological product code constructed from these codes with $\partial = \delta_1 \otimes \mathbb{1} + \mathbb{1} \otimes \delta_2$. Then, for any error~$E$ supported on fewer than $\lfloor (d_j-1)/2 \rfloor$ error bands~$E_a^i$, the error~$E'= U_{i}^\dagger E U_{i}$ is correctible using the recovery operator~$\mathcal{R}_j$, where $i,\ j \in \{1,2 \}$ such that~$i \ne j$.
\end{cor}

The above Corollary states that given a correctible error, as stated by Theorem~\ref{thm:errorbands}, conjugating that error by either decoding operator $U_{i}^{\dagger}$ will result in a correctible error on the remaining encoded states in the code~$\mathcal{C}_j$. 

We propose the following decoding algorithm. Given an encoded state in the traditional homological product code, we measure the syndromes of the code as specified by the row ($X$~type) and columns ($Z$~type) of the boundary operator~$\partial$. Now, given these measurement outcomes, we can map them onto syndrome outcomes for either of the two code, using the following procedure. Without loss of generality, suppose we would like to map them onto the syndromes of the code~$\mathcal{C}_2$. We know the modified boundary operator $\partial_{1,0} = (W_1^{-1} \otimes \mathbb{1}) \partial (W_1 \otimes \mathbb{1}) = \delta_{1,0} \otimes \mathbb{1} + \mathbb{1} \otimes \delta_2$ corresponds to $k_1$~logical states that are encoded in the code~$\mathcal{C}_2$. Specifically, the first $k_1n_2$ rows and columns of~$\partial_{1,0}$ will correspond to the stabilizers of the code~$\mathcal{C}_2$, see Eq.~\ref{eq:delta0matrix} for an example, replacing $\delta_{2,0}$ with~$\delta_2$. Suppose we measured a given stabilizer~$S_l$ of the original stabilizer code such that~$ES_l = (-1)^{b_l} S_lE$, that is $b_l \in \{0,1 \}$ records the stabilizer measurement outcome. Then since the encoding circuit is Clifford, and $S$~is Pauli, we can efficiently classically compute the form of the transformed syndrome~$S_l' = U_{1}^\dagger S_l U_{1}$. Moreover, $S_l'$ will keep the same commutation relation with the transformed error~$E'$, that is $E'S_l' = (-1)^{b_l} U_{1}^\dagger S_l E U_{1} = (-1)^{b_l} S_l'E'$. Therefore, we can use the transformed stabilizers $S_l'$ to determine the syndrome of~$E'$ in the code~$\mathcal{C}_2$. To find the appropriate recovery Pauli operator~$Q'$ we use the known decoder of~$\mathcal{C}_2$, and transform~$Q'$ back into a recovery operator for the original code by classically computing~$Q = U_{1} Q' U_{1}^{\dagger}$, which is again efficient since $U_{1}$ is a Clifford circuit.

We can then generalize the above method to address a build up of errors throughout the fault-tolerant process presented in Sec.~\ref{sec:partial}. Suppose, without loss of generality, we want to implement the logical gate~$G_2$ which is transversal for the code~$\mathcal{C}_2$. Then, we would start with unencoding~$\mathcal{C}_1$ by applying~$U_{1}^{\dagger} = V_1^\dagger \cdots V_t^{\dagger}$, where $V_i$ are the CNOT gates used in constructing the encoding unitary~$U_{1}$. After the application of each~$V_i^{\dagger}$, the code will be partially unencoded and the resulting boundary operator will be of the form~$\partial_{1,i} = \delta_{1,i} \otimes \mathbb{1} + \mathbb{1} \otimes \delta_2$, where $\delta_{1,i} = (V_i \cdots V_1) \delta_{1,0} (V_1^{\dagger} \cdots V_i^{\dagger})$\footnote{We have used an abuse of notation to denote $V_i$ as both the unitary CNOT gate and the binary representative for the CNOT. We feel it is clear from the context which operator we are referring to.}. If we assume that throughout the application of each~$V_i$ operator the generators of the code remain sparse, then we can measure these generators after each application~$V_i$ in order to address the errors that occurred during the application of that gate. The errors are then corrected by classically mapping the stabilizer generators~$\partial_{1,i}$ onto those of~$\partial_{1,0}$, and using, as outlined above, the decoder of~$\mathcal{C}_2$ to correct for the resulting errors.

A final remark on the stabilizer generators of~$\partial_{1,i}$. As stated above, if we were to measure them after every application of~$V_i, \ V_i^{\dagger}$, we would like them to remain sparse. In general, while the initial and final boundary operators~$\partial_{1,0}$ and~$\partial$ are certainly sparse, there will be no guarantee that the intermediary matrices remain sparse as well. However, for many common codes, such as topological codes, this can be achieved by choosing an appropriate encoding unitary. Roughly speaking, the idea is to build up the non-local logical operators of a topological code by growing the code at its boundary in a systematic manner~\cite{BSKFV11}. Moreover, if the stabilizers remain sparse at every time step, and the distance of each of the respective codes scale as~$n_i^{\alpha_i}$ for some positive power of~$\alpha_i$, then we can invoke the result Ref.~\cite{KP13} to prove the existence of a fault tolerance threshold. Moreover, as outlined in Ref.~\cite{KP13}, if the syndrome measurements are repeated then we may overcome measurement errors as well and still have a finite error threshold. Finally, repeated measurements may even be avoided if the underlying codes have single-shot correction properties~\cite{Bombin15b,Campbell18}. An explicit construction of a decoder applied to this setting remains an interesting open problem.

\section{Universal constructions}
\label{sec:Universal}

In this section, we explore an explicit example of a family of codes for implementing a universal set of fault-tolerant operations using the construction from Section~\ref{sec:partial}. We will focus on the Clifford $+ \ T$ universal gate set~\cite{BBC+95}. Let the code~$\mathcal{C}_1$ be the 2D~color code with distance~$d_1$ encoding a single logical qubit. The code has parameters~$[[ c_1 d_1^2, 1, d_1 ]]$, for a constant~$c_1$, and can implement any Clifford gate transversally~\cite{BM07, KB15}. The code~$\mathcal{C}_2$ will be composed of the gauge-fixed 3D~color code. That is, a particular choice of the 3D~color code where volume cell terms are of both $X$ and $Z$~type, while the face terms are only of $Z$~type. The resulting code has a transversal implementation of a~$T = \text{diag}(e^{-i \pi/8}, e^{i \pi/8})$, a non-Clifford gate and code parameters~$[[c_2 d_2^3, 1, d_2]]$, for a constant~$c_2$~\cite{Bombin15a,KB15}. However, due to the asymmetry in the number of $X$ and $Z$~stabilizer generators, arising from having to fix the face terms to be $Z$-type, the resulting code cannot be directly used in the single-sector homological product code construction. We present two alternative code constructions of~$\mathcal{C}_2$, one based on code padding, and one using two complementary copies and re-encoding them in the $[[4,2,2]]$ repetition code. 

\subsection{Code padding}
Suppose we have a CSS code~$\mathcal{C}$ that we want to use in a universal fault-tolerant implementation of a homological product code, and moreover assume without loss of generality there are more $Z$~generators than $X$~generators. In order to use the code~$\mathcal{C}$ in a homological product code construction, as presented, we must have the same number of independent $X$ and~$Z$ stabilizers. A rather simple alternative code that we can use is to pad the original code with extra ancillary qubits in the~$\ket{+}$ state, thus adding an extra set of single-qubit $X$~generators. The resulting code will have the same distance as the original code, where all non-trivial logical Pauli operators can be supported on the original qubits of the code. For example, in the case of the smallest gauge-fixed 3D color code, the 15-qubit Reed-Muller code, we can pad the code with an extra 6 $\ket{+}$~qubits, resulting in a~$[[21,1,3]]$ code. While this extra padding of qubits does not change the base code other than trivially alternating the number of underlying stabilizer generators, these generators will play a role in the homological product, via the initial entanglement present in the unencoded state, represented by~$\partial_0$.

Therefore, the universal scheme for implementing a set of fault-tolerant logical operations that can correct an arbitrary single qubit error will use the Steane and padded Reed-Muller codes, which are the smallest distance~3 versions of the 2D and padded 3D~color codes, respectively. The overall scheme will have coding parameters~$[[147, 1, 3^*]]$, where $3^*$ corresponds to the minimum fault-tolerant distance of the overall scheme, not necessarily the distance of the homological code itself. The stabilizer measurements will be of weight at most~15, see Appendix~\ref{app:ExBoundaryOps}. This is a large improvement over requiring measuring stabilizers of weight~28 in the concatenated scheme~\cite{JL14}, which leads to punitive threshold values and qubit overheads~\cite{CJL16,CJL17,CJ17}.

\subsection{Code doubling}

The process of code doubling was first presented in Refs.~\cite{Kitaev06, BTL10} for converting between Majorana fermion codes and stabilizer codes. We will outline the general logical procedure for any CSS code here, yet it can be generalized for arbitrary stabilizer code rather similarly.

Consider a CSS code~$\mathcal{A}^{(1)}$ of $n$~physical qubits with stabilizer generators~$S_{X_i}^{(1)} = \prod_{j \in \mathcal{X}_i} X_j^{(1)}$, where~$\mathcal{X}_i$ is a list of qubits in the support of~$S_{X_i}^{(1)}$ and the use of the superscript~$(1)$ will become clear shortly. Similarly, the $Z$~stabilizers are given by~$S_{Z_i}^{(1)} = \prod_{j \in \mathcal{Z}_i} Z_j^{(1)}$. Consider now a rotated version of~$\mathcal{A}^{(1)}$ where each of the $X$~stabilizers are replaced by $Z$~stabilizers and vice-versa, call this code~$\mathcal{A}^{(2)}$. More explicitly, the $X$ and $Z$~stabilizers of~$\mathcal{A}^{(2)}$ are given by: $S_{X_i}^{(2)} = \prod_{j \in \mathcal{Z}_i} X_j^{(2)}$,  $S_{Z_i}^{(2)} = \prod_{j \in \mathcal{X}_i} Z_j^{(2)}$. Therefore, the different superscripts represent different blocks of $n$~qubits. 

\begin{figure}
\centering
$
\Qcircuit @C=1em @R=.7em {
\lstick{\ket{\psi_1}} 	& \ctrl{2} & \qw & \qw & \targ & \qw\\
\lstick{\ket{\psi_2}} 	& \qw & \targ & \ctrl{1} & \qw & \qw \\
\lstick{\ket{0}}		& \targ & \qw & \targ & \qw & \qw \\
\lstick{\ket{+}}		& \qw & \ctrl{-2} & \qw & \ctrl{-3} & \qw
}
$
\vspace{2em}
\caption{Encoding circuit for the 4-qubit repetition code. The stabilizers generators of the code are: $X^{\otimes 4}, \ Z^{\otimes 4}$.}
\label{fig:4qubitEnc}
\end{figure}
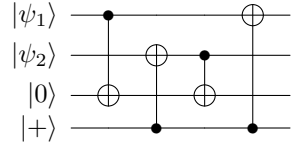

These two codes are then re-encoded into the $[[4,2,2]]$~repetition code, whose encoding circuit is given in Fig.~\ref{fig:4qubitEnc}. The third block of qubits will be initially prepared as~$\ket{0}^{\otimes n}$, while the fourth block will be prepared as~$\ket{+}^{\otimes n}$. Consider how the stabilizers are transformed under the action of the circuit in Fig.~\ref{fig:4qubitEnc}: 
\begin{align}
S_{X_i}^{(1)} = \prod_{j \in \mathcal{X}_i} X_j^{(1)} &\longrightarrow \prod_{j \in \mathcal{X}_i} X_j^{(1)} X_j^{(3)} \\
S_{X_i}^{(2)} = \prod_{j \in \mathcal{Z}_i} X_j^{(2)} &\longrightarrow \prod_{j \in \mathcal{Z}_i} X_j^{(2)} X_j^{(3)} \\
X_j^{(4)} &\longrightarrow X_j^{(1)}X_j^{(2)}X_j^{(3)}X_j^{(4)} \\
S_{Z_i}^{(1)} = \prod_{j \in \mathcal{Z}_i} Z_j^{(1)} &\longrightarrow \prod_{j \in \mathcal{Z}_i} Z_j^{(1)} Z_j^{(4)} \\
S_{Z_i}^{(2)} = \prod_{j \in \mathcal{X}_i} Z_j^{(2)} &\longrightarrow \prod_{j \in \mathcal{X}_i} Z_j^{(2)} Z_j^{(4)} \\
Z_j^{(3)} &\longrightarrow Z_j^{(1)}Z_j^{(2)}Z_j^{(3)}Z_j^{(4)}. 
\end{align}
Note that we can combine the mapped stabilizers with those of the repetition code in order to recognize a complete symmetry between the $X$ and $Z$~stabilizers of the code. While it may immediately follow from the presented construction, one can show that the above construction is equivalent to concatenating the $[[4,2,2]]$ code with the code~$\mathcal{A}^{(1)}$ and its rotated compliment~$\mathcal{A}^{(2)}$. Moreover, the distance of the new code will be twice that of the original code. Therefore, this concatenated code provides a code that can be used in the homological product code construction. 

Choosing the code~$\mathcal{A}^{(1)}$ to be the gauge-fixed 3D~color code along with its rotated compliment~$\mathcal{A}^{(2)}$, we can use these codes in conjunction with the 2D color code for the purposes of universal fault-tolerant computation via homological product codes. To perform any logical Clifford gate, it will be sufficient to decode the $[[4,2,2]]$ repetition code, followed by the decoding of~$\mathcal{A}^{(i)}$, for either or both $i \in \{ 1,2 \}$, depending on which codeblock one would like to apply the desired Clifford gate transversally.

In order to implement the $T$~gate fault-tolerantly, one would first decode the 2D~color code, as specified in Sec.~\ref{sec:partial}. At this point, one could not directly apply the non-Clifford gate transversally, as the two encoded codeblocks will still be further encoded in the $[[4,2,2]]$~code. However, one can then decode the~$[[4,2,2]]$ code by applying the circuit of Fig.~\ref{fig:4qubitEnc} in reverse. This will preserve the protection guaranteed by Theorem~\ref{thm:errorbands} as each block of 4-qubits will belong to the same error band, allowing the application of a transversal $T$~gate bookended by fault-tolerant operations. 

A final note regarding code doubling: since the stabilizers are symmetrized, the code will gain a transversal Hadamard gate. The logical result of the transversal Hadamard will be to implement logical Hadamard followed by logical SWAP between the two logical qubits. As such, for this particular operation, code doubling allows for a rapid implementation of this logical gate without having to decode one of the codes in the homological product. 

\section{Conclusion}
\label{sec:Conclusion}

This work introduces a method for implementing a set of logical gates using homological product codes, applicable to any set of CSS~codes. Namely, we show that if the underlying codes composing the homological product have complementary classes of transversal gates, then this scheme can be used to implement a fault-tolerant universal gate set. Moreover, if the underlying codes have stabilizer generators that are sparse, the construction will remain sparse, allowing for the implementation of a fault-tolerant gate set that does not require measurement of high-weight operators. This method is particularly interesting for the theory of quantum LDPC codes, where the hope would be to construct codes with good parameters and sets of transversal gates. If one were able to find such codes, this may lead to an alternative method for implementing universal fault-tolerant computation with constant overhead, in a similar manner to recent results that used special ancillary states to obtain a universal set of gates~\cite{Gottesman14,LTZ15,FGL18,FGL18b}. A recent result exploring the connection between homological product codes and single-shot error correction highlights a potential avenue for constructing codes with interesting transversal gates~\cite{Campbell18}, yet new constructions remain elusive.

The presented scheme relies on decoding one of the two codes composing the homological product, applying the transversal gate, and re-encoding. The encoding/decoding process may indeed spread errors in a dramatic way, yet due to the protection of the other code, the global operation remains fault-tolerant. If the encoder/decoder of each code preserves the sparsity of the code after each gate, then modified stabilizers may be measured during the encoding/decoding process, allowing for increased protection. Moreover, this should imply the existence of a finite error probability threshold due to the stabilizer generators always being low weight as long as the distance grows as a positive power of the number of qubits~\cite{KP13}.

\section*{Acknowledgments}
We thank Benjamin J. Brown, and Sam Roberts for insightful discussions on intermediary correction of errors during the fault-tolerant operations. We also thank Christopher Chamberland and John Preskill for comments during the development of this work and Sergey Bravyi for feedback on the initial manuscript. We acknowledge the support from the Walter Burke Institute for Theoretical Physics in the form of the Sherman Fairchild Fellowship as well as support from the Institute for Quantum Information and Matter~(IQIM), an NSF Physics Frontiers Center (NFS Grant PHY-1733907).

\bibliographystyle{unsrtnat}
\bibliography{bibtex_jochym}

\appendix

\section{Examples of boundary operators}
\label{app:ExBoundaryOps}

Boundary operator for the 7-qubit Steane code, each row and column have weight~4:
\begin{align}
\delta_7 =
\begin{pmatrix}
1 & 1 & 1 & 1 & 0 & 0 & 0 \\
1 & 1 & 0 & 0 & 1 & 1 & 0 \\
1 & 0 & 1 & 0 & 1 & 0 & 1 \\
1 & 0 & 0 & 1 & 0 & 1 & 1 \\
0 & 1 & 1 & 0 & 0 & 1 & 1 \\
0 & 1 & 0 & 1 & 1 & 0 & 1 \\
0 & 0 & 1 & 1 & 1 & 1 & 0 
\end{pmatrix}.
\end{align}
Note that for the Steane code, every non-trivial element of the stabilizer group is represented in the rows and columns, this will not hold in general for other codes.

A boundary operator for the padded 15-qubit Reed-Muller code, composed of 21~qubits:
\begin{align}
\delta_{15p} =
\begin{pmatrix}
011010011001011&111000\\
110000110011110&110010\\
101001010101101&101100\\
000011111111000&100100\\
100110010110011&011001\\
001100111100110&010010\\
010101011010101&001001\\
111111110000000&000000\\
100101101001011&000000\\
001111000011110&000010\\
010110100101101&000100\\
111100001111000&000100\\
011001100110011&000001\\
110011001100110&000010\\
101010101010101&000001\\
\hline
000000000000000&000000\\
000000000000000&000000\\
000000000000000&000000\\
000000000000000&000000\\
000000000000000&000000\\
000000000000000&000000
\end{pmatrix}.
\end{align}
We have visually split the matrix into two sets, the first 15~qubits and 6~ancillary qubits. The first 15~qubits are where the logical information is stored, while the extra 6~qubits represent the padded ancilla qubits that are prepared in the $\ket{+}$~state. Note that none of the $Z$~stabilizers, represented by the columns, have support on the last 6~qubits. It is fairly straightforward to check that~$\text{rank}(\delta_{15p}) = 10$. One can recover independent generators for the rows that correspond to the 15-qubit Reed-Muller code $X$~stabilizers on the first 15~qubits, and individual single-qubit $X$~generators on the last 6~qubits. One can also recover the independent 15-qubit Reed-Muller code $Z$~generators by considering the columns as well as a representative of the independent 6~gauge face generators in the last 6~columns, these correspond to fixing the gauge in the $Z$~basis.

The homological product boundary operator~$\partial = \delta_7 + \mathbb{1} + \mathbb{1} \otimes \delta_{15p}$ will have sparsity~15, that is every row and column in~$\partial$ will be of weight at most~$15$. This corresponds to the maximum weight operator one would have to measure for implementing the universal scheme on the homological product of these two codes, an improvement over the universal concatenated model~\cite{JL14} which would require measuring operators of weight~$28$. More importantly, in using higher distance versions of each of the codes, the concatenated model would require measuring operators whose weight will grow linearly with system size, as opposed to that of the homological construction which will remain constant-sized.

\end{document}